\documentclass[a4paper,english,notab]{eurocg26}

\usepackage{enumerate}
\newtheorem{observation}[theorem]{Observation}

\usepackage{cite} 
\usepackage{fnpct} 
\usepackage{microtype} 

\usepackage{nicematrix}
\usepackage[normalem]{ulem}
\usepackage{bm}
\usepackage{xspace}
\usepackage[disable]{todonotes}

\NewDocumentCommand{\red}{O{red}}{\textcolor{red}{#1}\xspace}
\NewDocumentCommand{\blue}{O{blue}}{\textcolor{blue}{#1}\xspace}
\NewDocumentCommand{\blockersrequired}{O{r}O{b}}{\ensuremath{\red[#1] \triangleright \blue[#2]}}
\DeclareMathOperator{\CH}{CH}
\newcommand{\G}{\mathcal{G}}
\renewcommand{\bowtie}{\ensuremath{\,\mathrm{bowtie}}}
\DeclareMathOperator{\cravat}{cravat}
\DeclareMathOperator{\skirt}{skirt}
\DeclareMathOperator{\pant}{pant}

\DeclareMathOperator{\necklace}{necklace}

\usepackage{etoolbox}
\newtoggle{long}\togglefalse{long} 

\iftoggle{long}{
	\NewDocumentEnvironment{prooflater}{m}{\begin{proof}}{\end{proof}}
	\NewDocumentEnvironment{proofsketch}{o +b}{}{}
	\newcommand{\restateref}[1]{}

	\NewDocumentEnvironment{statelater}{m}{}{}
	\NewDocumentCommand{\onlyShort}{+m}{}
	\NewDocumentCommand{\onlyLong}{+m}{#1}
}{
	\NewDocumentEnvironment{prooflater}{m +b}{ %
		\newcounter{#1-usages}\setcounter{#1-usages}{0}%
		\AtEndDocument{\ifnumequal{\value{#1-usages}}{0}{\todo[inline]{use prooflater `#1'}}{}}%
		\expandafter\global\expandafter\def\csname#1\endcsname{\stepcounter{#1-usages}\begin{proof}#2\end{proof}}%
	}{}
	\NewDocumentEnvironment{proofsketch}{O{Proof sketch.}}{\begin{proof}[#1]}{\end{proof}}
	\usepackage{apptools}
	\newcommand{\restateref}[1]{[\IfAppendix{\hyperref[#1]{$\star$}}{\hyperref[#1*]{$\star$}}]}

	\NewDocumentEnvironment{statelater}{m +b}{%
		\newcounter{#1-usages}\setcounter{#1-usages}{0}%
		\AtEndDocument{\ifnumequal{\value{#1-usages}}{0}{\todo[inline]{use statelater `#1'}}{}}%
		\expandafter\global\expandafter\def\csname#1\endcsname{\stepcounter{#1-usages}#2}%
	}{\ignorespacesafterend}
	\NewDocumentCommand{\onlyShort}{+m}{#1}
	\NewDocumentCommand{\onlyLong}{+m}{}
}

\let\oldrestatable\restatable
\def\restatable{\expandafter\oldrestatable}

\makeatletter
\pretocmd{\thmt@rst@storecounters}{\Hy@SaveLastskip}{}{}
\apptocmd{\thmt@rst@storecounters}{\Hy@RestoreLastskip}{}{}
\makeatother

\title{%
  Garment numbers\\ of bi-colored point sets in the plane%
}

\titlerunning{Garment numbers of bi-colored point sets in the plane}

\author[1]{Oswin Aichholzer}
\affil[1]{Institute of Algorithms and Theory, TU Graz, Austria\\
  \texttt{oswin.aichholzer@tugraz.at}}

\author[2]{Helena Bergold}
\affil[2]{Freie Universität Berlin, Germany\\
  \texttt{helena.bergold@fu-berlin.de}}

\author[3]{Simon D.~Fink}
\affil[3]{Algorithms and Complexity Group, TU Wien, Austria\\
  \texttt{sfink@ac.tuwien.ac.at}\\{\normalfont\footnotesize
  Vienna Science and Technology Fund (WWTF) grant [10.47379/ICT22029] and Austrian Science Fund (FWF) grant [10.55776/Y1329]}}

\author[4]{Maarten Löffler}
\affil[4]{Utrecht University, the Netherlands\\
  \texttt{m.loffler@uu.nl}}
\author[5,6]{Patrick Schnider}
\affil[5]{Department of Mathematics and Computer Science, University of Basel, Switzerland\\
  \texttt{patrick.schnider@unibas.ch}}
\affil[6]{Department of Computer Science, ETH Zürich, Switzerland\\
  \texttt{patrick.schnider@inf.ethz.ch}}

\author[7]{Josef Tkadlec}
\affil[7]{Computer Science Institute, Charles University, Prague, Czech Republic\\
  \texttt{josef.tkadlec@iuuk.mff.cuni.cz}\\{\normalfont\footnotesize
  GAČR project 25-17221S and Charles Univ.\ projects UNCE 24/SCI/008, PRIMUS 24/SCI/012}}

\authorrunning{O. Aichholzer, H. Bergold, S.\,D. Fink, M. Löffler, P. Schnider, J. Tkadlec} 

\makeatletter
\def\copyrightline{}
\gdef\@EventLongTitle{}
\gdef\@EventShortTitle{}
\makeatother

\begin{document}

\maketitle

\begin{abstract}
We consider colored variants of a class of geometric-combinatorial questions on $k$-gons and empty $k$-gons that have been started around 1935 by Erdős and Szekeres.
In our setting we have $n$ points in general position in the plane, each one colored either red or blue.
A \emph{structure} on $k$ points is a geometric graph where the edges are spanned by (some of) these points and is called \emph{monochromatic} if all $k$ points have the same color.
Already for $k=4$ there exist interesting open problems. Most prominently, it is still open whether for any sufficiently large bichromatic set
there always exists a \emph{convex} empty, monochromatic quadrilateral. In order to shed more light on the underlying geometry 
we study the existence of five different monochromatic structures that all use exactly 4 points of a bichromatic point set. 
We provide several improved lower and upper bounds on the smallest $n$ such that every bichromatic set of at least $n$ points contains (some of) those monochromatic structures.
\end{abstract}

\section{Introduction}

Around 1935 Erdős and Szekeres~\cite{erdosCombinatorialProblemGeometry1935} asked about the existence of a bound ${\mathrm{ES}}(k)$ such that any set of at least ${\mathrm{ES}}(k)$ points in general position in the plane contains a subset of $k$ points that are the vertices of a convex $k$-gon. This problem is also known as the ``Happy End Problem'', see e.g.~\cite{brassResearchProblemsDiscrete2005,hoffmanManWhoLoved1998} for further details. It started a whole area of geometric-combinatorial questions on $k$-gons and $k$-holes ($k$-gons which do not contain other points of the given set in their interior) in different settings and led to many beautiful results.

Variations when the points belong to different classes -- that are usually described as
\emph{colors} -- were introduced by Devillers, Hurtado, K{\'a}rolyi, and Seara~\cite{devillersChromaticVariantsErdosSzekeres2003}. Here we say that a (not necessarily convex or empty) structure is \emph{monochromatic} if all its points have the same color.
For sets of $n$ points in the plane with two colors (called \emph{bichromatic}), it has been shown that there are always at least
$\lceil \frac{n}{4} \rceil - 2$ monochromatic triangles with pairwise disjoint interiors~\cite{devillersChromaticVariantsErdosSzekeres2003}, and
at least
$\Omega(n^{4/3})$~\cite{pachMonochromaticEmptyTriangles2013} not-necessarily-disjoint empty monochromatic triangles.
On the other hand, Devillers et al. show that for $k \geq 5$ and any $n$ there are bichromatic sets of $n$ points where no empty monochromatic convex $k$-gon exists (Theorem 3.4 in~\cite{devillersChromaticVariantsErdosSzekeres2003}).
So somehow surprisingly the main remaining open problems in this setting are for sets of cardinality four.
Moreover, it has been shown~\cite{aichholzerLargeBichromaticPoint2010} that any sufficiently large (meaning with at least 2760 points) bichromatic set of points in the plane in general position determines an empty, monochromatic quadrilateral (which might not be convex).
While this statement sounds rather innocent its proof is non-trivial, and at the time when this result was shown the corresponding lower bound was just~18~\cite{aichholzerLargeBichromaticPoint2010}.
Moreover, the questions whether for a sufficiently large bichromatic set of points in the plane in general position there always exists a \emph{convex} empty, monochromatic quadrilateral is still open.
Here, the best known lower bound~is~46~\cite{koshelev2009erdosszekeresproblemrelatedproblems}.

Motivated by these results and with the goal to shed more light on the underlying structures we study the existence of combinations of different monochromatic structures on subsets of 4 points.
More precisely, we go beyond convex/non-convex polygons, 
and we introduce two new structures, the necklace and the bowtie; see~\Cref{fig:types} and the next section for a definition.
In total, we consider 5 different monochromatic structures that
all use exactly 4 points.
For combinations of these structures we study the following question:
What is the smallest $n$ such that every bichromatic set of $n$ points contains at least one of those empty monochromatic structures?


\section{Preliminaries}
Throughout this paper, let $P$ be a set of $n$ points in the plane in general position (that is, no three points in a set are collinear).
A \emph{structure} is a (closed) subset of the plane that is induced by 4 different points from $P$.
We consider five different structures, namely (see also~\Cref{fig:types}):

\begin{figure}[t]
  \centering
  \includegraphics[scale=1]{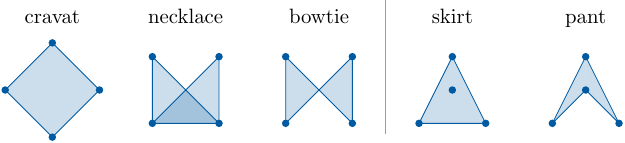}
  \caption{The five types of structures: cravat, necklace, bowtie, skirt, and pant.}
  \label{fig:types}
\end{figure}

\begin{enumerate}
    \item a \emph{cravat}\footnote{As in French ``à la cravate'', commonly also called necktie.}, defined by 4 points in convex position, is the convex hull of these 4 points.
    \item a \emph{necklace}, defined by 4 points in convex position, is the union of two triangles that share exactly two points that are consecutive on the convex hull.
    \item a \emph{bowtie}, defined by 4 points in convex position, is a self-intersecting 4-gon with those 4 points as vertices.
    \item a \emph{skirt}, defined by 4 points in non-convex position, is the convex hull of these 4 points.
    \item a \emph{pant}\footnote{The singular ``a pant'' is actually a viable alternative to ``pants'', especially in the \href{https://grammarphobia.com/blog/2015/11/pants.html}{fashion industry}, and combinatorially less confusing.}, defined by 4 points in non-convex position, is a simple 4-gon on these 4 points.
\end{enumerate}

Note that the first three structures have a convex hull of size 4, while the latter two have a convex hull of size 3.
Also note that each set of 4 points in convex position, depending on how they are ordered, induces 1~cravat, 4~necklaces, and 2~bowties; while in non-convex position it induces 1~skirt and 3~pants.

For a bichromatic point set $P$, we only consider structures that are monochromatic and for which there is no further point in $P$ of the same color within their boundary.
A point of the other color \textit{blocks} such a structure if it is within its boundary.
For a given bichromatic point set (colored, say, \red and \blue) and type of structure, we say that the \blue points \textit{block} the \red points if every \red structure contains at least one \blue point.
A structure that does not contain a point of the other color (and is thereby not blocked) is called \textit{empty}.
For a given set $S$ of structures, we say ``\emph{\red[$r$ red] points require \blue[$b$ blue] points}'' or shortly write \blockersrequired,
if every set of \red[$r$ red] points requires at least \blue[$b$ blue] points to block all \red structures in $S$.

We study whether all large enough bichromatic point sets contain an empty monochromatic structure.
Depending on the types of considered structure(s), we obtain different \emph{settings},
see \Cref{fig:settings} for an example.
Formally, given a set $S$ of structures, we define the \textit{GuARanteed Monochromatic Empty Not-quite-order Type number} (shorthand: Garment number) of bichromatic point sets, to be the smallest integer $n$, such that every bichromatic point set of size $n$ contains at least one empty monochromatic structure from $S$. (If no such integer exists, we set the garment number to $+\infty$.)
\Cref{fig:settings} illustrates that the Garment number of the set $S=\{\textrm{skirt, bowtie}\}$ is more than 7.
We denote this by $\G(\skirt\vee \bowtie)>7$.

\begin{figure}[t]
  \centering
  \includegraphics[width=\linewidth]{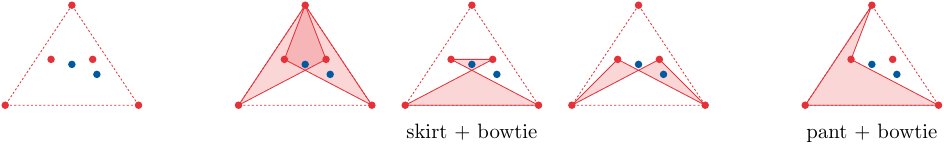}
  \caption{Left: A set of 7 points, 5 of them \red and 2 \blue. Middle: Each monochromatic skirt and bowtie contains a point of the other color, thus we have $\G(\skirt\vee \bowtie)>7$.
  Right: There is an empty pant, so the point set does not imply anything about the setting $\pant\vee\bowtie$.}
  \label{fig:settings}
\end{figure}

\begin{table}[t]
  \newcommand{\tablock}[2]{\Block[c]{}{\includegraphics[page=#1]{graphics/versions}\\#2}}
  \newcommand{\ours}[1]{\textcolor{orange}{#1}}
  \newcommand{\expl}[1]{\ensuremath{\bm{#1}}}
  \newcommand{\impl}[1]{\textit{#1}}
  \newcommand{\update}[2]{\sout{#1}\\\symb\ours{#2}}
  \newcommand{\uplow}[2]{\Block[c]{}{\diagbox{\def\symb{$>$}\symb#2}{\def\symb{$\leq$}\symb#1}}}
  \centering
  \begin{NiceTabular}{cc|c|c|c|c}[margin]
                &      & \tablock{6}{none} & \tablock{7}{cravat}                                 & \tablock{4}{necklace}                                     & \tablock{5}{bowtie}          \\ \Hline
  \tablock{6}{} & none   &  $\infty$       & \uplow{?}{\expl{46}}                                & \uplow{\ours{\expl{1508}}}{\ours{\expl{14}}}                               & \uplow{\ours{\impl{1508}}}{\ours{\impl{12}}}  \\ \Hline
  \tablock{2}{} & skirt  &  $\infty$       & \uplow{?}{\update{\impl{18}}{\expl{35}}}            & \uplow{\ours{\impl{1508}}}{\ours{\impl{12}}}                               & \uplow{\ours{\impl{1508}}}{\ours{\expl{12}}}  \\ \Hline
  \tablock{3}{} & pant   &  \expl{\infty}  & \uplow{\expl{2760}}{\update{\expl{18}}{\expl{22}}}  & \uplow{\update{\impl{2760}}{\expl{21}}}{\ours{\expl{12}}} & \uplow{\update{\impl{2760}}{\expl{11}}}{\ours{\expl{10}}}
  \end{NiceTabular}
  \quad
  \includegraphics[scale=1]{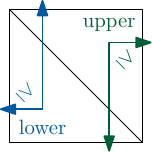}
  \smallskip
  \caption{
    The settings we consider in this paper, together with the upper and lower bounds on the number of points of a bichromatic point set in which an empty monochromatic structure always exists.
    Previously-known values are black, new results from this paper are orange.
    The results that are not a consequence of other results in this table are boldface and roman (i.e., non italic). Upper bounds are also upper bounds for settings below them or to their right, while lower bounds are also lower bounds for settings above them or to their left.
  }
  \label{fig:table}
\end{table}

\section{Upper Bounds}
We present upper bounds on Garment numbers in three different settings. They all include one structure with a convex hull of size 4, plus up to one structure with convex hull of size 3; see \Cref{fig:table}.
The first two of our upper bounds follow a similar three-step approach:

\begin{enumerate}
    \item By direct casework, we show that $\blockersrequired[r][b]$ for certain pairs $(\red[r],\blue[b])$ of small integers.
    \item Using an inductive argument, we show that if both $\blockersrequired[r-1][b-1]$ and $\blockersrequired[r][b]$, then $\blockersrequired[r+1][b+1]$.
    \item For large enough $n$, we locate a convex-like shape in $P$, and we find an empty structure there using $\blockersrequired[r+k][b+k]$, for suitable $k\ge 1$.
\end{enumerate}

In the rest of this section, we first present the inductive argument (\Cref{lem:induction}).
Then, we complete the proofs of the upper bounds for the settings $\{\textrm{pant, necklace}\}$ (\Cref{thm:necklace-ub}) and $\{\textrm{pant, bowtie}\}$ (\Cref{thm:bowtie-ub}). Finally, we prove \Cref{thm:necklace-only-ub} for the setting $\{\textrm{necklace}\}$.

Statements with a (clickable) $\star$ have their full proofs deferred to the appendix.\todo{full version}

\begin{figure}[t]
  \begin{subfigure}{.5\linewidth}
    \centering
    \includegraphics[page=1]{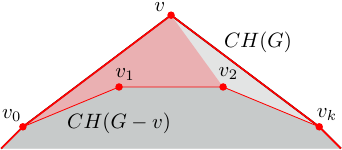}
    \caption{Case 1, where the convex hull does not shrink.}
    \label{fig:CHcuts-no-smaller}
  \end{subfigure}
  \begin{subfigure}{.5\linewidth}
    \centering
    \includegraphics[page=2]{graphics/CHcuts}
    \caption{Case 2, where the convex hull shrinks.}
    \label{fig:CHcuts-smaller}
  \end{subfigure}
  \caption{The two cases of \Cref{lem:induction}.}
  \label{fig:CHcuts}
\end{figure}
\begin{restatable}\restateref{lem:induction}{lemma}{lemInduction}\label{lem:induction}
  Let $S=\{\pant, \bowtie\}$ or $S=\{\pant, \necklace\}$.
  For all positive integers \red and \blue, it holds that
  $(\blockersrequired \;\land\; \blockersrequired[r-1][b-1]) \  \Rightarrow \  \blockersrequired[r+1][b+1]$.
\end{restatable}
\begin{proofsketch}
  We distinguish two cases depending on whether the boundary of the convex hull of the \red vertices contains a vertex whose removal yields a point set with at least that many \red points on the convex hull (Case 1), or not (Case 2); see \Cref{fig:CHcuts}.
\end{proofsketch}
\begin{prooflater}{proofInduction}
  Let $G=G_r\cup G_b$ be any point set, where $G_r$ is a set of \red[$r+1$ red] points and $G_b$ is a set of blue points. Suppose that all red structures in $S$ are blocked by $G_b$. We show that $|G_b|\ge b+1$.

  We distinguish two cases depending on whether the boundary $\CH(G_r)$ of the convex hull of the red vertices contains a vertex whose removal yields a point set with at least that many red points on the convex hull.


  \subparagraph*{Case 1:} $\exists v\in \CH(G_r) \colon |\CH(G_r-v)|\geq |\CH(G_r)|$.
  Let $v_0,v_1,\ldots,v_{k-1},v_k$ be the consecutive sequence of vertices on $\CH(G_r-v)$ that replaces the sequence $v_0,v,v_k$ on $\CH(G_r)$ when removing~$v$.
  Note that we have $k\geq2$ in this case.
  Observe that $P=v_0v_1v_2v$ is a pant -- an empty 4-gon with a reflex angle at $v_1$ in $G_r$; see \Cref{fig:CHcuts-no-smaller}.
  Since we assume \blockersrequired, we need \blue[$b$ blue] points to block $G_r-v$.
  Moreover, we need at least \blue[one blue] point to block the pant $P$.
  Thus $|G_b|\ge b+1$.

  \subparagraph*{Case 2:} $\forall v\in\CH(G_r) \colon |\CH(G_r-v)| < |\CH(G_r)|$.
  Observe that in this case we have $|\CH(G_r)|\geq4$.
  Let $v_0,u,v,v_k$ be four consecutive vertices of $\CH(G_r)$ that are replaced by the consecutive sequence $v_0,v_1,\ldots,v_{k-1},v_k$ on $\CH(G_r')$ with $k\geq0$ and $G_r'=G_r\setminus\{u,v\}$.
  Note that both $v_0,u,v$ and $u,v,v_k$ are triangles that contain no \red point in $G$, otherwise we would have case 1; see \Cref{fig:CHcuts-smaller}.
  The convex 4-gon $P=v_0v_1vu$ has to contain at least two \blue points,
  because if it contained only a single \blue point there would be a necklace or a bowtie on $P$ avoiding the \blue point.
  By assumption \blockersrequired[r-1][b-1], the polygon $\CH(G_r\setminus\{u,v\})$ contains at least \blue[$b-1$ blue] additional \blue[blue] points, so $|G_b|\ge b+1$.
\end{prooflater}

We now focus on the $\pant\vee\necklace$ setting.
The proof of the first following lemma is a relatively straightforward casework (see \Cref{fig:4pts-triangle-union}), while the second lemma can be shown by exhaustive enumeration on a computer.

\begin{figure}[t]
  \centering
  \includegraphics[width=\linewidth]{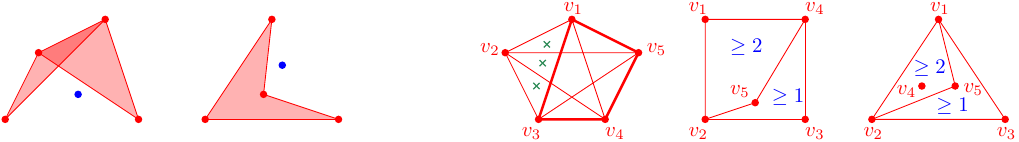}
  \caption{Left: We have \blockersrequired[ 4][2], regardless of whether the red quadrilateral is convex or not. Right: We have \blockersrequired[5][3], regardless of the size of the convex hull of the red points.}
  \label{fig:4pts-triangle-union}
\end{figure}

\begin{restatable}\restateref{lem:necklace-small}{lemma}{lemNecklaceSmall}\label{lem:necklace-small}
  For $S=\{\pant,\necklace\}$ we have \blockersrequired[4][2] and \blockersrequired[5][3].
\end{restatable}
\begin{prooflater}{proofNecklaceSmall}
  Regarding \blockersrequired[4][2], for contradiction consider a point set with 4 red points and just a single blue point. If the four red points are in convex position, there exists an empty red necklace, else there exists an empty red pant, see \Cref{fig:4pts-triangle-union} (left).
  Thus, \blockersrequired[4][2].

  Regarding \blockersrequired[5][3], for contradiction, suppose a set $G_r$ consisting of 5 red points can be blocked with just 2 blue points. We do a casework based on $|\CH(G_r)|$.

  \begin{enumerate}
      \item $|\CH(G_r)|=5$: Denote the red points by $v_1,\dots,v_5$ in the cyclical order. Since \blockersrequired[4][2], both blue points must lie inside $v_3v_4v_5v_1$, so $v_1v_2v_3$ is empty of blue points. Similarly, $v_2v_3v_4$ is empty. Thus a necklace on $v_1v_2v_3v_4$ is empty of blue points, a contradiction.
      \item $|\CH(G_r)|=4$: Denote the red points on the convex hull by $v_1,\dots,v_4$ in the cyclical order, and without loss of generality suppose that $v_5$ lies inside $v_2v_3v_4$.
      Since \blockersrequired[4][2], the convex quadrilateral $v_1v_2v_5v_4$ must contain at least two blue points, and (the interior-disjoint) pant on $v_2v_3v_4v_5$ must contain at least one more blue point, so at least 3 blue points are needed.
      \item $|\CH(G_r)|=3$: Denote the red points on the convex hull by $v_1,v_2,v_3$ and without loss of generality suppose that $v_4,v_5$ lie inside $v_1v_2v_3$ such that $v_2v_3v_5v_4$ is a convex quadrilateral with vertices in this order. As in the previous case, \blockersrequired[4][2] yields that triangle $v_1v_2v_5$ contains at least two blue points, and (the interior-disjoint) pant on $v_2v_3v_4v_5$ contains at least one more blue point, so at least 3 blue points are needed.\qedhere
  \end{enumerate}%
\end{prooflater}

\begin{restatable}\restateref{lem:6-island}{lemma}{lemSixIsland}\label{lem:6-island}
   Any set of $11$ points in general position contains either a convex 6-gon, or a convex 5-gon with at least 1 point inside of it.
\end{restatable}
\begin{prooflater}{proofSixIsland}
    For all $2\,334\,512\,907$ order types of size $11$~\cite{aichholzerAbstractOrderType2007} and each $k\le 11$, both the number of convex $k$-gons and the number of empty convex $k$-gons (also called $k$-holes) have been computed in previous work~\cite{aichholzerAbstractOrderType2007,order_webpage}.
    Using a computer program, we checked that each order type of size $11$ either has a positive number of 6-gons (and we are done), or it has strictly more $5$-gons than $5$-holes (thus some 5-gon contains at least 1 other point inside of it).
\end{prooflater}

We note that \Cref{lem:6-island} is minimal in the sense that there are sets of 10 points which contain neither a convex 6-gon nor a non-empty convex 5-gon (e.g.\ the double-chain with $5+5$~points).
This now allows us to show our first result.

\begin{theorem}\label{thm:necklace-ub}
  We have $\G(\pant\vee \necklace)\le 21$.
\end{theorem}
\begin{proof}
  Towards a contradiction, suppose there exists a bichromatic set of $n\ge 21$ points where each monochromatic pant and necklace are blocked.
  Consider the set $G$ of the first 21 points in $x$-sorted order.
  Suppose $G$ has $r$ red and $b$ blue points.
  W.l.o.g. assume $r>b$.
  By \Cref{lem:necklace-small} we have \blockersrequired[4][2] and \blockersrequired[5][3], thus by \Cref{lem:induction} we have \blockersrequired[12][10].
  If $r\geq 12$ then $b \leq 21-12=9<10$, a contradiction.
  Thus, we can assume that $r=11$ and $b=10$.

  By \Cref{lem:6-island}, the set of red points contains a subset $I$ of size $|I|\ge 6$ that has at least $|\CH(I)|\ge 5$ points on the boundary of its convex hull.
  By \Cref{lem:necklace-small,lem:induction}, $I$ contains a set $J$ of at least $|J|\ge |I|-2\ge 4$ blue points in the interior of its convex hull $\CH(I)$, and the set $J$ in turn contains at least $|J|-2\ge |I|-4$ red points in the interior of its convex hull $\CH(J)$.
  But this is a contradiction -- since $|\CH(I)|\ge 5$, at most $|I|-5$ red points may be in the interior of $\CH(I)$.
\end{proof}

To use induction for the upper bound in the $\pant\vee\bowtie$ setting, we again need a lemma for small point sets that uses some case work we defer to the appendix.\todo{full version}
\begin{restatable}\restateref{lem:bowtie-small}{lemma}{lemBowtieSmall}\label{lem:bowtie-small}
  For $S=\{\pant,\bowtie\}$ we have \blockersrequired[6][5] and \blockersrequired[7][6].
\end{restatable}

\begin{theorem}\label{thm:bowtie-ub}
  We have $\G(\pant\vee \bowtie)\le 12$.
\end{theorem}
\begin{proof}
    Given any bichromatic point set $P$ of $n\ge 12$ points, consider the set $P'$ of the 12 leftmost points.
    W.l.o.g. suppose $r\ge 6$ of them are red and $b=12-r$ are blue. If $r\ge 7$, then $b\le 5$ thus there is an unblocked red pant or bowtie by \Cref{lem:bowtie-small,lem:induction}.
    So suppose $r=b=6$.
    The boundary of $\CH(P')$ contains at least 3 vertices, w.l.o.g.\ at least 2 of them are blue. Removing them, we obtain a set of 6 red points and 4 blue points, thus again by \Cref{lem:bowtie-small,lem:induction} there is a red unblocked pant or bowtie.
\end{proof}



Using exhaustive enumeration on a computer, one can actually show the stronger statement that all bichromatic point sets of size 11 contain an empty pant or bowtie.
Together with~\Cref{obs:lb} from below, this establishes the following corollary.

\begin{restatable}\restateref{cor:bowtie-equal}{corollary}{corBowtieEqual}\label{cor:bowtie-equal}
    We have  $\G(\pant\vee \bowtie)= 11$.
\end{restatable}
\begin{prooflater}{proofBowtieEqual}
    For contradiction, suppose that there exists a set $P$ of 11 points with all pants and bowties blocked. Suppose there are $r$ \red points and $b$ \blue points and w.l.o.g. suppose $r>b$. If $r\ge 7$ then $b\le 4$ and we get a contradiction with \blockersrequired[7][6] (\Cref{lem:bowtie-small}). Thus, suppose $r=6$ and $b=5$.

    We claim that $P$ must have a fully \red convex hull of size 3, and a second \blue convex layer of size 4 (and then 3 more \red points and 1 \blue point).
    Once this claim is proved, we conclude as follows: All order types with 11 points have been generated~\cite{aichholzerAbstractOrderType2007}.
    Using a computer we checked all bichromatic point sets that fulfill these properties, and we showed that each of them contains an empty pant or bowtie.

    It remains to prove the claim.
    Consider the boundary $\CH(P)$ of the convex hull of $P$. If $P$ contains a \blue point $v$, then the set $P\setminus\{v\}$ has six \red points and 4 \blue points, thus it contradicts \blockersrequired[6][5] (\Cref{lem:bowtie-small}).
    So all points in $\CH(P)$ are \red.
    
    If $|\CH(P)|\ge 4$, then the set $P\setminus \CH(P)$ has at most 2 \red points and precisely 5 \blue points, thus it contradicts \blockersrequired[5][3] (with colors reversed, \Cref{lem:necklace-small}). Thus $|\CH(P)|=3$.

    Let $Q=P\setminus \CH(P)$. Then $Q$ has $5$ \blue points and 3 \red points. If the boundary $\CH(Q)$ of the convex hull of $Q$ contains a red point $v$ then $Q\setminus v$ has 5 \blue points and 2 \red points, thus again contradicting \blockersrequired[5][3] (with colors reversed). Thus all points in $\CH(Q)$ (the ``second convex layer'') are blue.

    Finally, 
    from the beginning of the proof of \Cref{lem:bowtie-small} we know that for a set $R$ of 5 \red points we have \blockersrequired[5][4], unless $|\CH(R)|=4$. Reversing the colors and noting that $Q$ has 5 \blue points and only 3 \red points, we conclude that $Q$ must satisfy $|\CH(Q)|=4$.
    This concludes the proof of the claim.
\end{prooflater}

\NewDocumentCommand\counterexample{mmo}{
  \begin{subfigure}{.45\linewidth}
    \centering
    \IfNoValueTF{#3}{
      \includegraphics[width=\linewidth]{graphics/counterexamples/n#1_c2_no_mc_#2}
      \caption{#1 points with no #2.}
      \label{fig:lb-#2}
    }{
      \includegraphics[width=\linewidth]{graphics/counterexamples/n#1_c2_no_mc_#2_#3}
      \caption{#1 points with no #2 or #3.}
      \label{fig:lb-#2-#3}
    }
  \end{subfigure}
}
\begin{figure}[p]
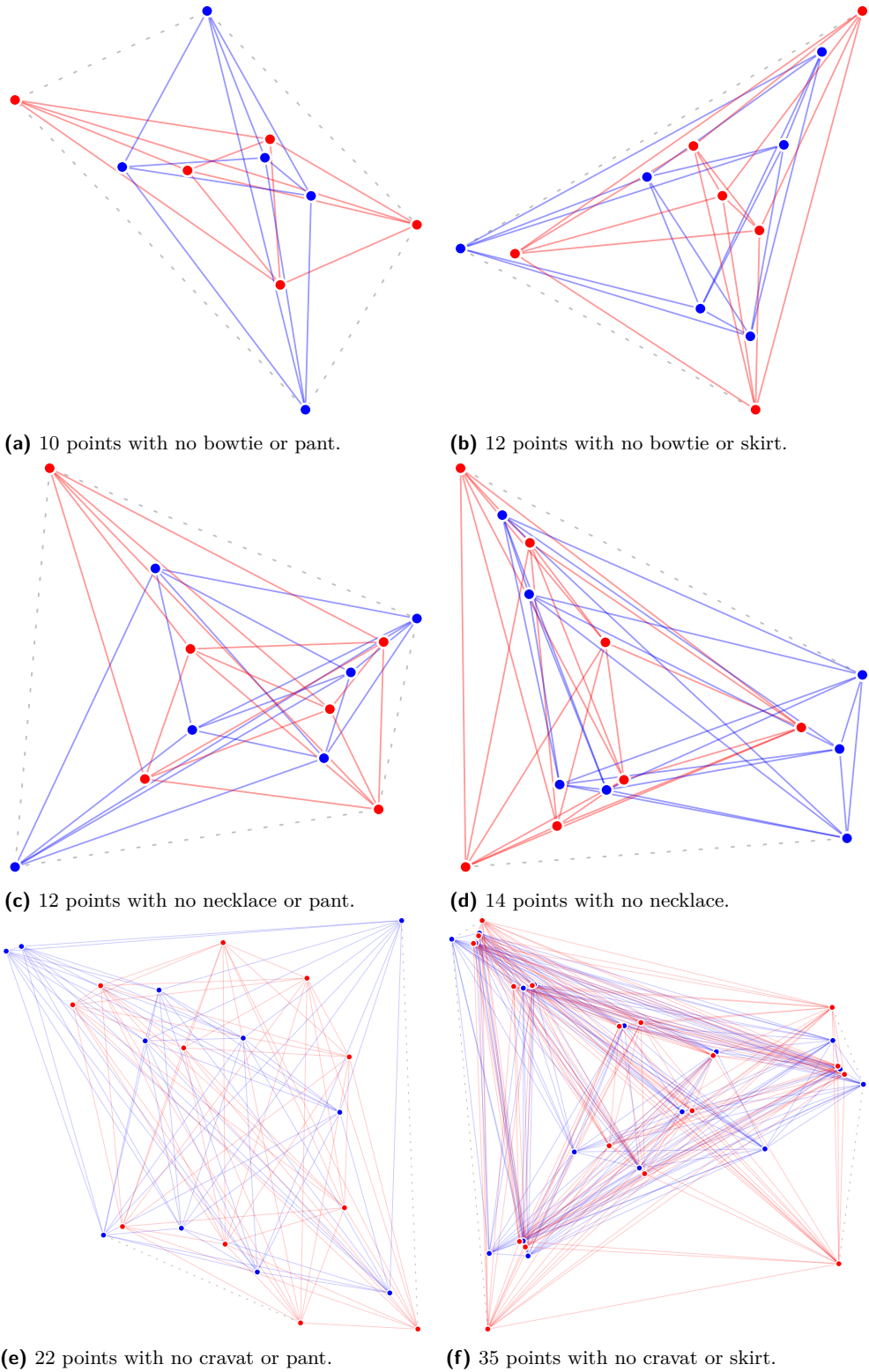

  \centering
  \counterexample{10}{bowtie}[pant]
  \counterexample{12}{bowtie}[skirt]
  \counterexample{12}{necklace}[pant]
  \counterexample{14}{necklace}
  \counterexample{22}{cravat}[pant]
  \counterexample{35}{cravat}[skirt]
  \caption{Our lower bound constructions, all avoiding certain empty monochromatic structures.}
  \label{fig:counterexamples}
\end{figure}

Our third upper bound also finds an empty structure within a large enough convex shape.

\begin{restatable}\restateref{thm:necklace-only-ub}{theorem}{thmNecklaceUb}\label{thm:necklace-only-ub}
  We have $\G(\necklace)\le 1508$.
\end{restatable}
\begin{proofsketch}
    We first use the Erd\H{o}s–Szekeres theorem~\cite{erdosCombinatorialProblemGeometry1935} on convex $k$-gons
    to find a bicolored subset which consists of, say, at least five more blue than red points.
    Then we use a result on counting the number of interior disjoint convex blue 4-gons to show that
    there are not enough red points to double-block them all.
    Thus at least one blue convex 4-gon generates a blue necklace.    
\end{proofsketch}

\section{Lower bounds}\label{sec:lb}
Our lower bounds follow from examples of point sets without empty monochromatic structures.
More specifically, a lower bound of the form $\G(S)>x$ follows from a configuration of $x$ points that contains no empty monochromatic $S$ shown in \Cref{fig:counterexamples}.

\begin{observation}\label{obs:lb}
  The following lower bounds hold:\\
  \setlength{\tabcolsep}{.5ex}
  \begin{tabular}{@{\hspace{2em}} rll @{\hspace{2em}} rll @{\hspace{2em}} rll}
      (a) & $\G(\bowtie\vee\pant)$   & $>10$ & 
      (c) & $\G(\necklace\vee\pant)$ & $>12$ & 
      (e) & $\G(\cravat\vee\pant)$   & $>22$   
      \\
      (b) & $\G(\bowtie\vee\skirt)$  & $>12$ & 
      (d) & $\G(\necklace)$          & $>14$ & 
      (f) & $\G(\cravat\vee\skirt)$  & $>35$   
  \end{tabular}
\end{observation}

The exact point positions, together with code for verifying their correctness, can be found on GitHub\cite{github}.
The set of 22 points was found by students Severin Kann, J\"urgen Pammer, and Matthias Platzer; the set of 35 points by Christian Payer and Karlheinz Wohlmuth.


\section{Conclusion}

We consider this work as a starting point of a more systematic exploration of blocking the various 4-point structures of bichromatic point sets. As seen in \Cref{fig:table}, many open questions remain.
Next to the already-known open question of further closing the gap between lower and upper bounds for the pant and cravat setting, the main open part is finding any upper bounds for settings without pants.
The main difference between settings with pants and those with only skirts, is that blocking a pant requires at least two points of the other color, while a skirt requires only one.

\subsection*{Acknowledgments}
We would like to thank Philipp Kindermann as organizer of GG Week 2024 in Trier where this work was initiated.

\bibliography{references}

\appendix

\section{Upper Bounds with Pants}
\lemInduction*
\label{lem:induction*}
\proofInduction

\lemNecklaceSmall*
\label{lem:necklace-small*}
\proofNecklaceSmall

\lemSixIsland*
\label{lem:6-island*}
\proofSixIsland

\lemBowtieSmall*
\label{lem:bowtie-small*}
\begin{proof}
First, note that since each necklace contains a bowtie, by \Cref{lem:necklace-small} we have \blockersrequired[4][2] and \blockersrequired[5][3].
Moreover, we claim that if a set $P$ of $|P|=5$ points has either 3 or 5 points on the boundary $\CH(P)$ of its convex hull, then even 4 blockers are required. For contradiction, suppose that 3 blockers suffice.
\begin{enumerate}
\item Case $|\CH(P)|=3$. Label the points as in \Cref{fig:5ptCH3}.

\begin{figure}[h]
  \centering
  \includegraphics{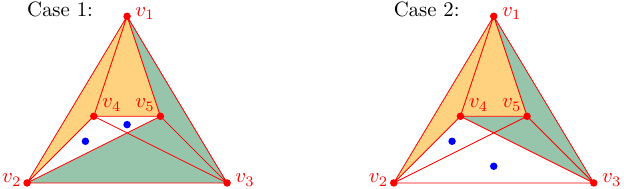}
  \caption{\blockersrequired[5][4] for 5 red points with 3 of them on the convex hull.}
  \label{fig:5ptCH3}
\end{figure}

The quadrilateral $\{v_2,v_3,v_4,v_5\}$ needs at least two blue points for both of its bowties to be blocked, and there are two ways to do this (up to symmetry). In both cases we can find two disjoint unblocked bowties, which cannot be blocked simultaneously by the last blue point (see \Cref{fig:5ptCH3}), a contradiction.

\item Case $|\CH(P)|=5$. Label the points as in \Cref{fig:5ptCH5}.

  \begin{figure}[h]
  \centering
  \includegraphics[width=\textwidth]{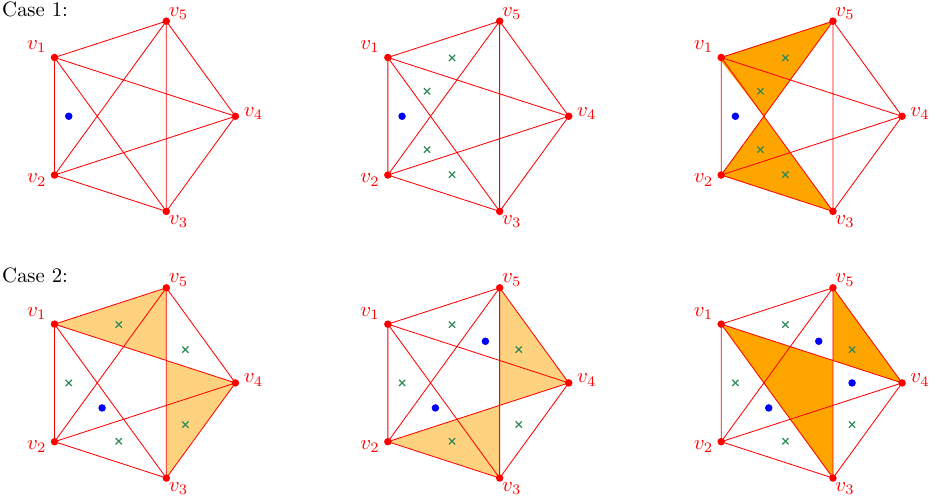}
  \caption{\blockersrequired[5][4] for 5 red points in convex position.}
  \label{fig:5ptCH5}
\end{figure}

Suppose that there is a blue point in the cell incident to $v_1v_2$. This is Case 1 in \Cref{fig:5ptCH5}. Then both the quadrilaterals spanned by $\{v_2,v_3, v_4, v_5\}$ and $\{v_1,v_3, v_4, v_5\}$ need the remaining two blue points to be blocked. In particular, there cannot be an additional blue point in the triangles $\{v_1, v_2, v_3\}$ and $\{v_1, v_2, v_5\}$ (cells where there are no blue points are marked with green crosses in \Cref{fig:5ptCH5}). But then the quadrilateral $\{v_1,v_2, v_3, v_4\}$ contains an unblocked bowtie.

We can thus assume that no cell incident to a convex hull edge contains a blue point. Then one of the cells incident only to a vertex must contain a blue point, say, the one incident to $v_2$. This is Case 2 in \Cref{fig:5ptCH5}. In order to block the bowtie in the quadrilateral $\{v_1,v_3, v_4, v_5\}$ that is not incident to the convex hull edge, we must place another blue point in a cell incident to a vertex, without loss of generality incident to $v_5$. Analogously, to block the same bowtie in the quadrilateral $\{v_2,v_3, v_4, v_5\}$ we must place a blue point without loss of generality in the cell incident to $v_4$. But now we have placed 3 blue points and the quadrilateral $\{v_1,v_3, v_4, v_5\}$ still contains an unblocked bowtie. This shows that if the 5 red points are in convex position, we need at least 4 blue points to block all red bowties.
\end{enumerate}

With this in mind, we proceed to show \blockersrequired[6][5].
Let $P$ be a red set of $|P|=6$ points and for contradiction, suppose that all bowties and pants can be blocked using only 4 blue points.
We again do a casework based on the number of points on the boundary $\CH(P)$ of the convex hull of $P$:
\begin{enumerate}
\item Case $|\CH(P)|=6$. Label the points as in \Cref{fig:6ptCH}. Blocking $v_2v_3v_4v_5v_6$ requires 4 points, so $v_1v_2v_6$ is empty. Likewise for $v_1v_2v_3$, thus there is an empty bowtie $v_1v_3v_2v_6$.

\begin{figure}[h]
  \centering
  \includegraphics[width=\textwidth]{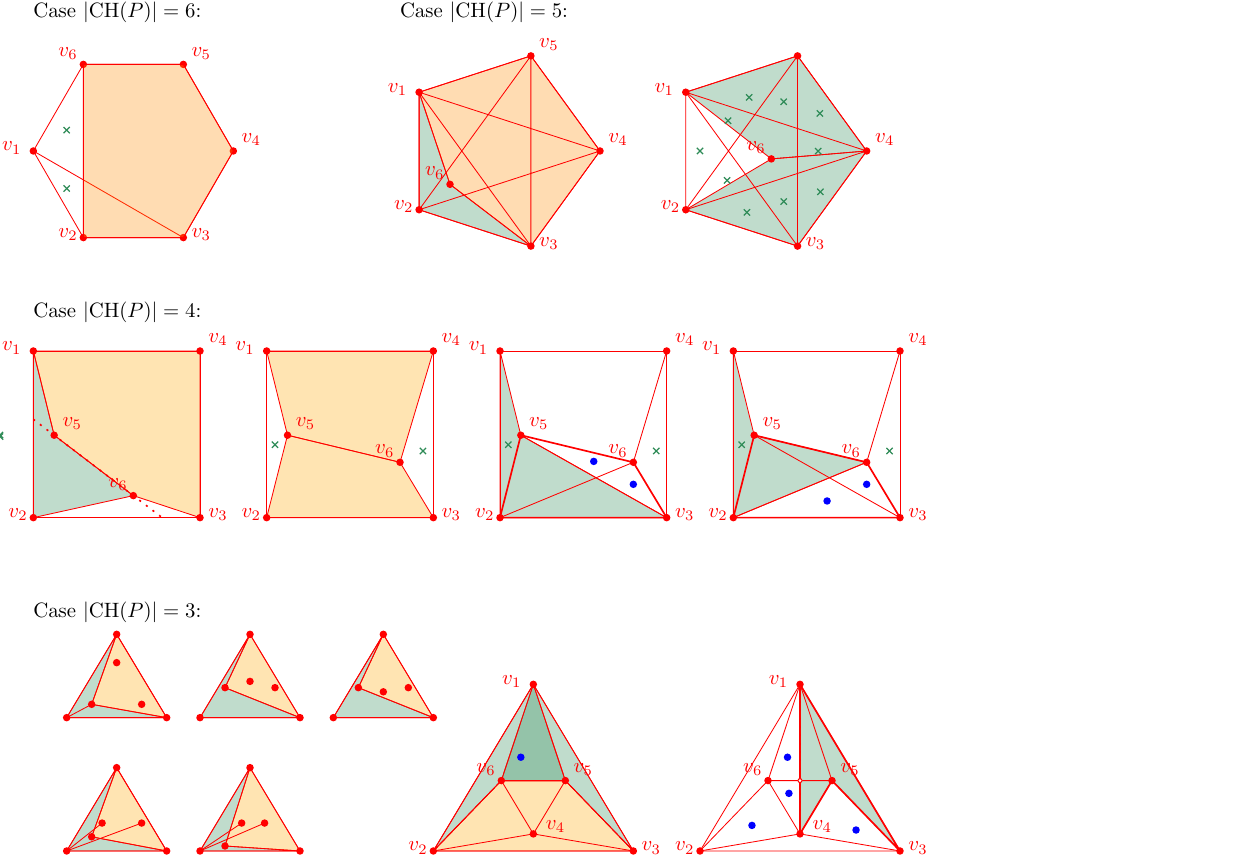}
  \caption{\blockersrequired[6][5] depending on the size of the convex hull.}
  \label{fig:6ptCH}
\end{figure}

\item Case $|\CH(P)|=5$. Label the convex hull vertices by $v_1v_2v_3v_4v_5$ and the 6th point by~$v_6$. If $v_6$ lies in one of the ``boundary'' triangles $v_iv_{i+1}v_{i+2}$ (with indices taken mod~5), then the convex hull of $P$ can be split into a pant $v_iv_{i+1}v_{i+2}v_6$ and the remaining interior-disjoint pentagon, thus $1+4=5$ blockers are required, see \Cref{fig:6ptCH}. So suppose $v_6$ lies in the central pentagon, see \Cref{fig:6ptCH}.

Split $v_1v_2v_3v_4v_5$ into quadrilaterals $v_1v_6v_4v_5$, $v_2v_3v_4v_6$ and a triangle $v_1v_2v_6$. Since \blockersrequired[4][2], each quadrilateral requires 2 blockers, thus the triangle is empty. By the same logic, $v_2v_3v_6$ is empty, a contradiction.

\item Case $|\CH(P)|=4$. Label the convex hull vertices by $v_1v_2v_3v_4$ and the other two points by $v_5$, $v_6$. If the line $v_5v_6$ crosses two adjacent sides of $v_1v_2v_3v_4$ then $v_1v_2v_3v_4$ contains a pant and an interior-disjoint pentagon, thus $1+4=5$ blockers are required, see \Cref{fig:6ptCH}. So wlog suppose that $v_5v_6$ crosses $v_1v_2$ and $v_3v_4$ so that quadrilaterals $v_1v_5v_4v_6$ and $v_2v_3v_5v_6$ are convex. Each of them requires two blockers, so triangles $v_1v_2v_5$ and $v_3v_4v_6$ are both empty. There are two ways to block the two bowties in $v_2v_3v_5v_6$ using 2 blue points (up to symmetry). In both cases, we can find an empty quadrilateral, see \Cref{fig:6ptCH}. (Note that $v_1v_2v_3v_5$ is either a pant, or a convex quadrilateral, which contains a bowtie.)

\item Case $|\CH(P)|=3$. In this case, there are 6 possible order types. In five of those 6 cases, it is possible to split the convex hull into a pant and a set of 5 points with convex hull of size 3, thus requiring $1+4=5$ blockers. In the final, 6th case label the points as in \Cref{fig:6ptCH}.

A quadrilateral $v_2v_3v_5v_6$ requires 3 blockers, so the final blocker must be in $v_1v_6v_5$ to block both pants $v_1v_2v_6v_5$ and $v_1v_6v_5v_3$. Similarly, there is a blocker in $v_2v_4v_6$ and in $v_3v_5v_4$. The 4th blocker must guarantee that each of $v_1v_6v_4v_5$ and $v_2v_4v_5v_6$ has 2 blockers, thus it must be in $v_4v_5v_6$. Let $p=v_1v_4\cap v_5v_6$. To block both bowties in $v_1v_6v_4v_5$, wlog assume that there is one blocker in each of $v_1v_6p$ and $v_4pv_6$. Then $v_1v_4v_5v_3$ is an empty pant, a contradiction.
\end{enumerate}

Finally, we show \blockersrequired[7][6]. The argument is as in the proof of \Cref{lem:induction}.
Let $P$ be a set of $|P|=7$ points and let $h_1,h_2,\dots,h_k$ be the $k\in\{3,4,5,6,7\}$ vertices on the boundary of its convex hull $\CH(P)$. For contradiction, suppose that $P$ can be blocked using only 5 points.

\begin{enumerate}
\item If some ``boundary triangle'' $h_ih_{i+1}h_{i+2}$ (with indices taken mod $k$) contains at least one other point $p\in P$, then by \blockersrequired[6][5] the convex hull of $P\setminus\{h_{i+1}\}$ requires at least 5 blockers, and there is an interior-disjoint pant that requires 1 additional blocker (see Case 1 in \Cref{lem:induction}). Note that this takes care of cases $k\in\{3,4\}$.
\item If all ``boundary triangles'' $h_ih_{i+1}h_{i+2}$ are empty, we aim to find $i$ such that  $P'=P\setminus\{h_i,h_{i+1}\}$ is a convex 5-gon. This suffices, since $P'$ then requires 4 blockers, and there is an interior-disjoint convex 4-gon with side $h_ih_{i+1}$ that requires additional 2 blockers (see Case 2 in \Cref{lem:induction}). We distinguish cases based on $k$: 
\begin{enumerate}
\item When $k=7$, any $i$ yields a 5-gon $P'$.

\item When $k=6$, consider a main diagonal $h_1h_4$ and wlog assume that the 7th point $v_7$ lies inside $h_1h_2h_3h_4$. Then $i=2$ yields a 5-gon.
\item When $k=5$, the last two points $v_6$, $v_7$ both lie in the central pentagon. The line $v_6v_7$ cuts two non-consecutive sides, wlog $h_1h_2$ and $h_3h_4$. Then $i=2$ yields a 5-gon.\qedhere
\end{enumerate}

\end{enumerate}
\end{proof}

\corBowtieEqual*
\label{cor:bowtie-equal*}
\proofBowtieEqual

\section{Upper Bound for only Necklace}
An island $I$ of a point set $P$ is a subset of points of $P$ such that $Conv(I) \cap P = I$, that is, the island $I$ contains all points of $P$ that lie inside the convex hull of $I$. Note that if $I$ contains any of the structures shown in~\Cref{fig:types} than this is independent of the points from $P$ outside of $I$.
If an island $I$ is bicolored with $r$ red and $b$ blue points we say that $I$ is unbalanced if $|r-b| \geq 5$, that is, the larger color class has at least 5 points more than the smaller color class.

\begin{lemma}\label{lem:unbalanced}
  Any bicolored set of at least $1508$ points contains an unbalanced island.
\end{lemma}

\begin{proof}
    Let $P$ be any bichromatic point set of size $n = 1508$ with $r$ red and $b$ blue points. We use the Erd\H{o}s–Szekeres theorem~\cite{erdosCombinatorialProblemGeometry1935} to find an unbalanced subset $I$ of~$P$.
    As already mentioned in the introduction, this theorem says that for every integer $k$ there is a number ${\mathrm{ES}}(k)$ such that any set of at least ${\mathrm{ES}}(k)$ points in general position in the plane contains a subset of $k$ points that are the vertices of a convex $k$-gon.
    Recently, Mojarrad and Vlachos~\cite{MojarradVlachos} provided an upper bound ${\mathrm{ES}}(k) \leq \binom{2k-5}{k-2} - \binom{2k-8}{k-3}+2$ for any $k \geq 6$.
    Thus ${\mathrm{ES}}(9)\leq 1508$ which implies that any set of at least $1508$ points contains a convex 9-gon. (For larger $k$, stronger upper bounds exist, see e.g.~\cite{sukErdosSzekeresConvexPolygon2017}.)
    
    Let $I$ be the island of $P$ that has this 9-gon as the boundary of the convex hull of $I$.
    Let $r'$ and $b'$ be the number of red and blue points of~$I$, respectively. If $|r'-b'| \geq 5$ we already have our unbalanced island. Otherwise we assume w.l.o.g.\ that there are at least as many blue points as red points and that $b'=r'+i$ for some $0 \leq i < 5$. Now if there are at least $5-i$ red points on the boundary of the convex hull of $I$ we remove $5-i$ of these red points from $I$ to obtain an unbalanced island $I'$. Else there are at least $9-(4-i)=5+i$ blue points on the boundary of the convex hull of $I$. Remove these $5+i$ blue points from $I$ to obtain the island $I'$. $I'$ now contains $b''=b'-(5-i)=r'+i-(5-i)=r'-5$ blue points and is thus an unbalanced island, as $r'=b''+5$.
\end{proof}

Note that it is straightforward to generalize \Cref{lem:unbalanced} to $k$-unbalanced islands by considering convex $(2k-1)$-gons in the Erd\H{o}s–Szekeres theorem.

\thmNecklaceUb*
\label{thm:necklace-only-ub*}
\begin{proof}
By \Cref{lem:unbalanced} any bicolored set of at least $1508$ points contains an unbalanced island~$I$. 
W.l.o.g.\ assume that $I$ consists of $b$ blue and $r$ red points and that there are more blue than red points, that is, $b \geq r+5$. To prove our statement we use Theorem~2 of~\cite{quadrilaterals2019} which says that any set $T$ of $m$ points in general position in the plane always admits a family of $\lfloor \frac{m-3}{2} \rfloor$ interior disjoint convex 4-holes. Let $T$ be the set of blue points of $I$ with $m=b$. Then there are at least $\lfloor \frac{b-3}{2} \rfloor$ interior disjoint blue convex 4-gons which do not contain any other blue point in their interior. To place two red points inside each such convex blue 4-gon we would need at least $2\lfloor \frac{b-3}{2} \rfloor \geq b-4$ red points. As $r \leq b-5$ there exists at least one convex blue quadrilateral in $I$ (and thus $P$) which contains at most one red point and therefor this quadrilateral forms a blue necklace.
\end{proof}

\end{document}